\documentclass[11pt,leqno]{article}
\usepackage{amsmath, amsfonts, amsthm, amssymb, amscd}
\usepackage{setspace}
\usepackage{cite}
\usepackage[pdftex]{graphicx}

\newtheorem{thm}{Theorem}
\newtheorem{cor}[thm]{Corollary}
\newtheorem{lem}[thm]{Lemma}
\newtheorem{prop}[thm]{Proposition}

\newtheorem{defn}[thm]{Definition}

\newtheorem{rmk}[thm]{Remark}

\newcommand{\tensor}{\otimes}

\newcommand{\p}{\prime}

\newcommand{\R}{\mathbb{R}}
\newcommand{\Z}{\mathbb{Z}}
\newcommand{\Q}{\mathbb{Q}}
\newcommand{\C}{\mathbb{C}}

\newcommand{\isom}{\cong}

\DeclareMathOperator{\Gal}{Gal}
\DeclareMathOperator{\Det}{Det}

\newcommand{\beeta}[1]{{\beta^{(#1)}}}
\newcommand{\beeeta}[2]{{\beta_{#1}^{(#2)}}}
\newcommand{\ex}[1]{{x^{(#1)}}}

\newcommand{\ol}{\overline}

\newcommand{\coop}[1]{\mathfrak{C}(#1)}

\newcommand{\Pj}{\mathbb{P}}

\newcommand{\surj}{\twoheadrightarrow}
\newcommand{\inj}{\hookrightarrow}

\newcommand{\comment}[1]{}

\newcommand{\bi}{\begin{itemize}}
\newcommand{\ei}{\end{itemize}}
\newcommand{\ben}{\begin{enumerate}}
\newcommand{\een}{\end{enumerate}}
\newcommand{\be}{\begin{equation}}
\newcommand{\ee}{\end{equation}}
\newcommand{\bea}{\begin{eqnarray}}
\newcommand{\eea}{\end{eqnarray}}
\newcommand{\bal}{\begin{align}}
\newcommand{\eal}{\end{align}}
\newcommand{\ba}{\begin{array}}
\newcommand{\ea}{\end{array}}
\newcommand{\nn}{\nonumber}

\newcommand{\sm}{{\begin{scriptsize}\mbox{sm}\end{scriptsize}}}

\newcommand{\I}{I}
\newcommand{\J}{J}

\DeclareMathOperator{\GL}{GL}

\DeclareMathOperator{\M}{M}

\begin{document}
\title{Robust Coin Flipping}
\author{Gene S. Kopp\\\texttt{gkopp@umich.edu} \and John D. Wiltshire-Gordon\thanks{It is our pleasure to thank Tom Church, for helping simplify our original proof of algebraicity of mystery-values; L\'aszl\'o Babai, for providing guidance with respect to publication; L\'aszl\'o Csirmaz, for discussing secret-sharing with us; Victor Protsak, for pointing us to Lind's article \cite{lind}; and Matthew Woolf, Nic Ford, Vipul Naik, and Steven J. Miller for reading drafts and providing helpful comments.  We are also grateful to several anonymous referees for their suggestions.}\\\texttt{johnwg@umich.edu}}

%\title{When can several people simulate a private\\ random source for someone else?}

%\date{\today}

%EXTRA INFO FOR amsart
%\subjclass[2010]{94A60, 94A62, 60C05, 68Q87 (primary), 15A69, 15B48, 14M12, 55R80, 11K60 (secondary)}
%\subjclass[2010]{94A60 (cryptography), 94A62 (Authentication and secret sharing), 15A69 (multilinear algebra), 15B48 (positive matrices etc.), 14M12 (determinental varieties...not quite right),  60C05 (combinatorial probability), 68Q87 (probability in computer science), 55R80 (Discriminantal varieties, configuration spaces), 11K60 (Diophantine approximation)} WHICH ONES SHOULD WE INCLUDE?

\maketitle

\begin{abstract}
Alice seeks an information-theoretically secure source of private random data.  Unfortunately, she lacks a personal source and must use remote sources controlled by other parties.  Alice wants to simulate a coin flip of specified bias $\alpha$, as a function of data she receives from $p$ sources; she seeks privacy from any coalition of $r$ of them.  We show: If $p/2 \leq r < p$, the bias can be any rational number and nothing else; if $0 < r < p/2$, the bias can be any algebraic number and nothing else.  The proof uses projective varieties, convex geometry, and the probabilistic method.    Our results improve on those laid out by Yao, who asserts one direction of the $r=1$ case in his seminal paper \cite{yao}.  We also provide an application to secure multiparty computation.
\end{abstract}

%\tableofcontents

\[  \]\section{Introduction}

Alice has a perfectly fair penny---one that lands heads exactly 50\% of the time.  Unfortunately, the penny is mixed in with a jar of ordinary, imperfect pennies.  The truly fair penny can never be distinguished from the other pennies, since no amount of experimentation can identify it with certainty.  Still, Alice has discovered a workable solution.  Whenever she needs a fair coin flip, she flips all the pennies and counts the Lincolns; an even number means heads, and an odd number means tails.

Alice's technique is an example of ``robust coin flipping.''  She samples many random sources, some specified number of which are unreliable, and still manages to simulate a desired coin flip.  Indeed, Alice's technique works even if the unreliable coin flips somehow fail to be independent.

Bob faces a sort of converse problem.  He's marooned on an island, and the nearest coin is over three hundred miles away.  Whenever \emph{he} needs a fair coin flip, he calls up two trustworthy friends who don't know each other, asking for random equivalence classes modulo two.  Since the sum of the classes is completely mysterious to either of the friends, Bob may safely use the sum to make private decisions.

Bob's technique seems similar to Alice's, and indeed we shall see that the two predicaments are essentially the same.  We shall also see that the story for biased coin flips is much more complex.

\subsection{Preliminaries and Definitions}
Informally, we think of a random source as a (possibly remote) machine capable of sampling from certain probability spaces.  Formally, a \textbf{random source} is a collection $\mathcal{C}$ of probability spaces that is closed under quotients.  That is, if $X \in \mathcal{C}$ and there is a measure-preserving map\footnote{A measure-preserving map (morphism in the category of probability spaces) is a function for which the inverse image of every measurable set is measurable and has the same measure.  Any measure-preserving map may be thought of as a quotient ``up to measure zero.''} $X \to Y$, then $Y \in \mathcal{C}$.  Random sources are partially ordered by inclusion: We say that $\mathcal{C}$ is \textbf{stronger than} $\mathcal{D}$ iff $\mathcal{C} \supset \mathcal{D}$.

The quotients of a probability space $X$ are precisely the spaces a person can model with $X$.  For example, one can model a fair coin with a fair die: Label three of the die's faces ``heads" and the other three ``tails."  Similarly, one can model the uniform rectangle $[0,1]^2$ with the uniform interval $[0,1]$: Take a decimal expansion of each point in $[0,1]$, and build two new decimals, one from the odd-numbered digits and one from the even-numbered digits.\footnote{In fact, this defines an isomorphism of probability spaces between the rectangle and the interval.}  Thus, forcing $\mathcal{C}$ to be closed under quotients is not a real restriction; it allows us to capture the notion that ``a fair die is more powerful that a fair coin.''\footnote{It would also be natural (albeit unnecessary) to require that $\mathcal{C}$ is closed under finite products.}

We define an \textbf{infinite random source} to be one that contains an infinite space.\footnote{An infinite space is one that is not isomorphic to any finite space.  A space with exactly 2012 measurable sets will always be isomorphic to a finite space, no matter how large it is as a set.}  A \textbf{finite random source}, on the other hand, contains only finite probability spaces.  Further, for any set of numbers $\mathbb{S}$, we define an \textbf{$\mathbb{S}$-random source} to be one which is forced to take probabilities in $\mathbb{S}$.  That is, all the measurable sets in its probability spaces have measures in $\mathbb{S}$.

Sometimes we will find it useful to talk about the strongest random source in some collection of sources.  We call such a random source \textbf{full-strength} for that collection.  For instance, a full-strength finite random source can model any finite probability space, and a full-strength $\mathbb{S}$-random source can model any $\mathbb{S}$-random source.

In practice, when $p$ people simulate a private random source for someone else, they may want to make sure that privacy is preserved even if a few people blab about the data from their random sources or try to game the system.  Define an \textbf{$r$-robust} function of $p$ independent random variables to be one whose distribution does not change when the joint distribution of any $r$ of the random variables is altered.  Saying that $p$ people simulate a random source $r$-robustly is equivalent to asserting that the privacy of that source is preserved unless someone learns the data of more than $r$ participants.  Similarly, to simulate a random source using $p$ sources, at least $q$ of which are working properly, Alice must run a $(p-q)$-robust simulation.

By a \textbf{robust} function or simulation, we mean a $1$-robust one.

We use $J$ to denote the all-ones tensor of appropriate dimensions.  When we apply $J$ to a vector or hypermatrix, we always mean ``add up the entries.''

\subsection{Results}\label{results}
This paper answers the question ``When can a function sampling from $p$ independent random sources be protected against miscalibration or dependency among $p-q$ of them?'' (Alice's predicament), or equivalently, ``When can $p$ people with random sources simulate a \emph{private} random source for someone else\footnote{Later, we give an application to secure multiparty computation in which the output of the simulated random source has no single recipient, but is utilized by the group without any individual gaining access; see Section \ref{application}.} in a way that protects against gossip among any $p-q$ of them?'' (Bob's predicament).  In the first question, we assume that at least $q$ of the sources are still functioning correctly, but we don't know which.  In the second question, we assume that at least $q$ of the people keep their mouths shut, but we don't know who.  In the terminology just introduced, we seek a $(p-q)$-robust simulation.

Consider the case of $p$ full-strength finite random sources.  We prove: If $1 \leq q \leq p/2$, the people may simulate any finite $\mathbb{Q}$-random source and nothing better; if $p/2 < q < p$, they may simulate any finite $\overline{\mathbb{Q}}$-random source and nothing better.  The proof uses projective varieties, convex geometry, and the probabilistic method.  We also deal briefly with the case of infinite random sources, in which full-strength simulation is possible, indeed easy (see Appendix \ref{infapp}).

\subsection{Yao's robust coin flipping}
Our work fits in the context of secure multiparty computation, a field with roots in A. C. Yao's influential paper \cite{yao}.  
%He provides solutions to several natural problems about cooperative computation between people that don't completely trust each another.  
In the last section of his paper, entitled ``What cannot be done?'', Yao presents (a claim equivalent to) the following theorem:

\begin{thm}[A. C. Yao]
Alice has several finite random sources, and she wants to generate a random bit with bias $\alpha$.  Unfortunately, she knows that one of them may be miscalibrated, and she doesn't know which one.  This annoyance actually makes her task impossible if $\alpha$ is a transcendental number.
\end{thm}

\noindent It does not not suffice for Alice to just program the distribution $(\alpha \hspace*{8pt} 1-\alpha)$ into one of the random sources and record the result; this fails because she might use the miscalibrated one!  We require---as in our jar of pennies example---that Alice's algorithm be robust enough to handle unpredictable results from any single source.%  She will need to use several random sources to get the job done.

Unfortunately, Yao provides no proof of the theorem, and we are not aware of any in the literature.  Yao's theorem is a special case of the results we described in the previous section.

\section{Simulating finite random sources}

The following result is classical.

\begin{prop}\label{diecon}
If $p$ players are equipped with private $d$-sided dice, they may $(p-1)$-robustly simulate a $d$-sided die.
\end{prop}
\begin{proof}
We provide a direct construction.  Fix a group $G$ of order $d$ (such as the cyclic group $\Z/d\Z$).  The $i^{\rm th}$ player uses the uniform measure to pick $g_i \in G$ at random.  The roll of the simulated die will be the product $g_1g_2 \cdots g_p$.

It follows from the $G$-invariance of the uniform measure that any $p$-subset of \be \{g_1, g_2,...,g_p, g_1g_2 \cdots g_p\} \ee is independent!  Thus, this is a $(p-1)$-robust simulation.
\end{proof}

\noindent For an example of this construction, consider how Alice and Bob may robustly flip a coin with bias $2/5$.  Alice picks an element $a \in \Z/5\Z$, and Bob picks an element $b \in \Z/5\Z$; both do so using the uniform distribution.  Then, $a,b,$ and $a+b$ are pairwise independent!  We say that the coin came up heads if $a+b \in \{0,1\}$ and tails if $a+b \in \{2,3,4\}$.

This construction exploits the fact that several random variables may be pairwise (or $(p-1)$-setwise) independent but still dependent overall.  In cryptology, this approach goes back to the one-time pad.  Shamir \cite{shamir} uses it to develop secret-sharing protocols, and these are exploited in multiparty computation to such ends as playing poker without cards \cite{poker, game}.

\begin{cor}\label{qcon}
If $p$ players are equipped with private, full-strength finite $\Q$-random sources, they may $(p-1)$-robustly simulate a private, full-strength finite $\Q$-random source for some other player.
\end{cor}
\begin{proof}
Follows from Proposition \ref{diecon} because any finite rational probability space is a quotient of some finite uniform distribution.
\end{proof}

\subsection{Cooperative numbers}

We define a useful class of numbers.

\begin{defn}
If $p$ people with private full-strength finite random sources can robustly simulate a coin flip with bias $\alpha$, we say $\alpha$ is \textbf{p-cooperative}.  We denote the set of $p$-cooperative numbers by $\coop{p}$.
\end{defn}

\noindent The ability to robustly simulate coin flips of certain bias is enough to robustly simulate any finite spaces with points having those biases, assuming some hypotheses about $\coop{p}$ which we will later see to be true.

\begin{lem}
Suppose that, if $\alpha, \alpha^\p \in \coop{p}$ and $\alpha < \alpha^\p$, then $\alpha/\alpha^\p \in \coop{p}$.  If $p$ people have full-strength finite random sources, they can robustly simulate precisely finite $\coop{p}$-random sources.
\end{lem}
\begin{proof}
Clearly, any random source they simulate must take $p$-cooperative probabilities, because any space with a subset of mass $\alpha$ has the space $(\alpha \hspace*{8pt} 1-\alpha)$ as a quotient.

In the other direction, consider a finite probability space with point masses \be (\ba{cccc} \alpha_1 & \alpha_2 & \cdots & \alpha_n \ea) \ee in $\coop{p}$.  Robustly flip a coin of bias $\alpha_1$.  In the heads case, we pick the first point.  In the tails case, we apply induction to robustly simulate \be (\ba{ccc} \alpha_2/(1-\alpha_1) & \cdots & \alpha_n/(1-\alpha_1) \ea). \ee  This is possible because $1-\alpha_1 \in \coop{p}$ by symmetry, and so the ratios $\alpha_i/(1-\alpha_1) \in \coop{p}$ by assumption.
\end{proof}

\subsection{Restatement using multilinear algebra}
Consider a $\{$heads, tails$\}$-valued function of several independent finite probability spaces that produces an $\alpha$-biased coin flip when random sources sample the spaces.  If we model each probability space as a stochastic vector---that is, a nonnegative vector whose coordinates sum to one---we may view the product probability space as the Kronecker product of these vectors.  Each entry in the resulting tensor represents the probability of a certain combination of outputs from the random sources.  Since the sources together determine the flip, some of these entries should be marked ``heads,'' and the rest ``tails.''

For instance, if we have a fair die and a fair coin at our disposal, we may cook up some rule to assign ``heads'' or ``tails'' to each combination of results:

\be 
\left(
\begin{array}{c}
\frac16 \\ 
\vspace*{-9pt} \\ 
\frac16 \\ 
\vspace*{-9pt} \\ 
\frac16 \\ 
\vspace*{-9pt} \\  
\frac16 \\ 
\vspace*{-9pt} \\
\frac16 \\ 
\vspace*{-9pt} \\ 
\frac16
\end{array}
\right)
\tensor
\left(
\begin{array}{cc}
\frac12 & \frac12
\end{array}
\right)
=\left(
\begin{array}{cc}
 \frac1{12} & \frac1{12} \\
 \vspace*{-9pt} \\
 \frac1{12} & \frac1{12} \\
 \vspace*{-9pt} \\
 \frac1{12} & \frac1{12} \\
 \vspace*{-9pt} \\ 
 \frac1{12} & \frac1{12} \\
 \vspace*{-9pt} \\
 \frac1{12} & \frac1{12} \\
 \vspace*{-9pt} \\
 \frac1{12} & \frac1{12}
\end{array}
\right)
\longrightarrow
\left(
\begin{array}{cccccc}
 H & T \\
 \vspace*{-9pt} \\
 H & T \\
 \vspace*{-9pt} \\
 T & H \\
 \vspace*{-9pt} \\
 H & T \\
 \vspace*{-9pt} \\
 T & H \\
\vspace*{-9pt} \\
 T & H
\end{array}
\right)
\ee
If we want to calculate the probability of heads, we can substitute $1$ for $H$ and $0$ for $T$ in the last matrix and evaluate
\be 
\left(
\begin{array}{cccccc}
\frac16 & \frac16 & \frac16 & \frac16 & \frac16 & \frac16
\end{array}
\right)
\left(
\begin{array}{cc}
 1 & 0 \\
 1 & 0 \\
 0 & 1 \\
 1 & 0 \\
 0 & 1 \\
 0 & 1
\end{array}
\right)
\left(
\begin{array}{c}
\frac12 \\
\vspace*{-9pt} \\
\frac12
\end{array}
\right)
 = \frac12.
\ee
This framework gives an easy way to check if the algorithm is robust in the sense of Yao.  If one of the random sources is miscalibrated (maybe the die is a little uneven), we may see what happens to the probability of heads:
\be 
\left(
\begin{array}{cccccc}
\frac1{12} & \frac1{10} & \frac16 & \frac14 & \frac1{15} & \frac13
\end{array}
\right)
\left(
\begin{array}{cc}
 1 & 0 \\
 1 & 0 \\
 0 & 1 \\
 1 & 0 \\
 0 & 1 \\
 0 & 1
\end{array}
\right)
\left(
\begin{array}{c}
\frac12 \\
\vspace*{-9pt} \\
\frac12
\end{array}
\right)
 = \frac12.
\ee
It's unaffected!  In fact, defining 
\be 
A\left(\ex{1}, \ex{2}\right) = \ex{1} \left(
\begin{array}{cc}
 1 & 0 \\
 1 & 0 \\
 0 & 1 \\
 1 & 0 \\
 0 & 1 \\
 0 & 1
\end{array}
\right)
 \ex{2}^\top,
\ee
we see that letting 
$\beeta{1} = \left(
\begin{array}{cccccc}
\frac16 & \frac16 & \frac16 & \frac16 & \frac16 & \frac16
\end{array}
\right)$ 
and 
$\beeta{2} = \left(
\begin{array}{cc}
\frac12 & \frac12
\end{array}
\right)$ 
gives us

\begin{eqnarray}
A\left(\ex{1}, \beeta{2}\right) & = & \frac12 \nn \\
A\left(\beeta{1}, \ex{2}\right) & = & \frac12
\end{eqnarray}

\noindent for all $\ex{1}$ and $\ex{2}$ of mass one.  These relations express Yao's notion of robustness; indeed, changing at most one of the distributions to some other distribution leaves the result unaltered.  As long as no two of the sources are miscalibrated, the bit is generated with probability $1/2$.

If $\alpha$ denotes the bias of the bit, we may write the robustness condition as
\begin{eqnarray}
A\left(\ex{1}, \beeta{2}\right) & = & \alpha J\left(\ex{1}, \beeta{2}\right) \nn \\
A\left(\beeta{1}, \ex{2}\right) & = & \alpha J\left(\beeta{1}, \ex{2}\right)
\end{eqnarray}
since the $\beeta{i}$ both have mass one.  (Here as always, $J$ stands for the all-ones tensor of appropriate dimensions.)  These new equations hold for all $\ex{i}$, by linearity.  Subtracting, we obtain
\begin{eqnarray}
0 & = & (\alpha J - A)\left(\ex{1}, \beeta{2}\right) \nn \\
0 & = & (\alpha J - A)\left(\beeta{1}, \ex{2}\right)
\end{eqnarray}
which says exactly that the bilinear form $\left( \alpha J - A \right)$ is degenerate, i.e., that
\be 
\Det(\alpha J - A) = 0.\footnote{If the matrix $(\alpha J - A)$ is not square, this equality should assert that all determinants of maximal square submatrices vanish. }
\ee
These conditions seem familiar: Changing the all-ones matrix $J$ to the identity matrix $I$ would make $\alpha$ an eigenvalue for the left and right eigenvectors $\beeta{i}$.  By analogy, we call $\alpha$ a \emph{mystery-value} of the matrix $A$ and the vectors $\beeta{i}$ \emph{mystery-vectors}.  Here's the full definition:

\begin{defn}
A $p$-linear form $A$ is said to have \textbf{mystery-value} $\alpha$ and corresponding \textbf{mystery-vectors} $\beeta{i}$ when, for any $1 \leq j \leq p$,
\be 
0 = (\alpha J - A) \left( \beeta{1}, \ldots , \beeta{j-1} , \ex{j} ,\beeta{j+1}, \ldots , \beeta{p} \right) \mbox{ for all vectors $\ex{j}$.}
\ee
We further require that $J(\beeta{i}) \not = 0$.
\end{defn}

\noindent We will see later that these conditions on $\left( \alpha J - A \right)$ extend the notion of degeneracy to multilinear forms in general.  This extension is captured by a generalization of the determinant---the hyperdeterminant.\footnote{Hyperdeterminants were first introduced in the $2 \times 2 \times 2$ case by Cayley \cite{cayley}, and were defined in full generality and studied by Gelfand, Kapranov, and Zelevinsky \cite[Chapter 14]{GKZ}.}  Hyperdeterminants will give meaning to the statement $\Det (\alpha J - A) = 0 $, even when $A$ is not bilinear.

This organizational theorem summarizes our efforts to restate the problem using multilinear algebra.

\begin{thm}
A function from the product of several finite probability spaces to the set $\{H, T\}$ generates an $\alpha$-biased bit robustly iff the corresponding multilinear form has mystery-value $\alpha$ with the probability spaces as the accompanying mystery-vectors.
\end{thm}

\noindent We may now show the equivalence of robustness and privacy more formally.  Privacy requires that $(\alpha J - A)\left(\otimes \beeta{i} \right)$ remains zero, even if one of the distributions in the tensor product collapses to some point mass, that is, to some basis vector.\footnote{That is, the simulated bit remains a ``mystery'' to each player, even though she can see the output of her own random source.}  This condition must hold for all basis vectors, so it extends by linearity to Yao's robustness.

\subsection{Two players}
The case $p=2$ leaves us in the familiar setting of bilinear forms.  

\begin{prop}[Uniqueness]
Every bilinear form has at most one mystery-value.
\end{prop}
\begin{proof}
Suppose $\alpha$ and $\alpha^\p$ are both mystery-values for the matrix $A$ with mystery-vectors $\beeta{i}$ and $\beeta{i}^\p$, respectively.  We have four equations at our disposal, but we will only use two:

\begin{eqnarray}
A\left(\ex{1} \hspace{1.73pt}, \rule[0pt]{0pt}{12pt} \hspace{1.73pt} \beeta{2} \right) & = & \alpha \nn \\
A\left(\beeta{1}^\p, \rule[0pt]{0pt}{12pt} \ex{2} \right) & = & \alpha^\p
\end{eqnarray}
We observe that a compromise simplifies both ways:
\be 
\alpha = A\left(\beeta{1}^\p, \beeta{2} \right) = \alpha^\p,
\ee
so any two mystery-values are equal.
\end{proof}

\begin{cor}
Two players may not simulate an irrationally-biased coin.
\end{cor}
\begin{proof}
Say the $\{0,1\}$-matrix $A$ has mystery-value $\alpha$.  Any field automorphism $\sigma \in \Gal(\C/\Q)$ respects all operations of linear algebra, so $\sigma(\alpha)$ is a mystery-value of the matrix $\sigma(A)$.  But the entries of $A$ are all rational, so $\sigma(A)=A$.  Indeed, $\sigma(\alpha)$ must also be a mystery-value of $A$ itself.  By the uniqueness proposition, $\sigma(\alpha) = \alpha$.  Thus, $\alpha$ is in the fixed field of every automorphism over $\Q$ and cannot be irrational.
\end{proof}

\begin{thm}\label{2rat}
%If exactly two participants have finite random sources, then the strongest source they may simulate is a finite $\Q$-random source.
$\coop{2} = \Q \cap [0,1]$.  Two people with finite random sources can robustly simulate only $\Q$-random sources; indeed, they can already simulate a full-strength finite $\Q$-random source if they have full-strength finite $\Q$-random sources.
\end{thm}
\begin{proof}
The previous corollary shows that no probability generated by the source can be irrational, since it could be used to simulate an irrationally-biased coin.  The other direction has already been shown in Corollary \ref{qcon}.
%If Alice and Bob each pick an element of $G$ privately, then the composition of the two elements specifies some point in $\Omega$ via the quotient map.  This point is independent of either group element due to the invariance of the Haar measure, so the privacy condition is satisfied.
\end{proof}

\begin{prop}\label{ratcon}
If $p$ people have full-strength finite $\Q$-random sources, they may $(p-1)$-robustly simulate any finite $\Q$-random source.
\end{prop}
\begin{proof}
Follows from Proposition \ref{diecon} just as the constructive direction of Theorem \ref{2rat} does.
\end{proof}

\subsection{Three or more players: what can't be done}
Even if three or more players have private finite random sources, it remains impossible to robustly simulate a transcendentally-biased coin.  The proof makes use of algebraic geometry, especially the concept of the dual of a complex projective variety.  We describe these ideas briefly in Appendix \ref{geoapp}.  For a more thorough introduction, see \cite[Lec. 14, 15, 16]{harris} or \cite[Ch. 1]{GKZ}.

%\su%bsubsection{Robustly simulating a transcendentally-biased coin flip is impossible}
Let $A$ be a rational multilinear functional of format $n_1 \times \cdots \times n_p$ (see Section \ref{format}), and let $X$ be the Segre variety of the same format.  Set $n := n_1 \cdots n_p - 1$, the dimension of the ambient projective space where $X$ lives.  In what follows, we prove that $A$ has algebraic mystery-values. This is trivial when $A$ is a multiple of $J$, and for convenience we exclude that case.

\begin{prop}\label{tangency}
Let $A$ have mystery-value $\alpha$ with corresponding mystery-vectors $\beeta{i}$.  Define $\beta = \otimes \beeta{i}$, and let $\mathfrak{B}$ denote the hyperplane of elements of $(\Pj^n)^\ast$ that yield zero when applied to $\beta$.  Now $(\mathfrak{B}, \left( \alpha J - A \right))$ is in the incidence variety $W_{X^\vee}$ (see Section \ref{geoapp1}).
\end{prop}
\begin{proof}
By the biduality theorem \ref{biduality}, the result would follow from the statement,
\be \mbox{``The hyperplane     } \left\{\hspace{.02in} x \hspace{.015in} : \hspace{.025in} (\alpha J - A)(x) = 0 \right\} \mbox{       is tangent to $X$ at $\beta$.''} \ee
But this statement is true by the partial derivatives formulation (Definition \ref{pderiv}) of the degeneracy of $\left( \alpha J - A \right)$.
\end{proof}

\noindent It is a standard fact (see \textit{e.g.} \cite[p. 6]{mumford}) that any variety has a stratification into locally closed smooth sets.  The first stratum of $X^{\vee}$ is the Zariski-open set of smooth points of the variety.  This leaves a subvariety of strictly smaller dimension, and the procedure continues inductively.  Equations for the next stratum may be found by taking derivatives and determinants.

Since $X^{\vee}$ itself is defined over $\mathbb{Q}$, it follows that each of its strata is as well.  We conclude that there must be some subvariety $S \subseteq X^{\vee}$, defined over $\Q$, that contains $\left( \alpha J - A \right)$ as a smooth point.

\begin{thm}
Any mystery-value of $A$ must be an algebraic number.
\end{thm}
\begin{proof}
Let $A' = \alpha J - A$, and let $\ell$ be the unique projective line through $A$ and $J$.  Let $\mathbb{A}$ be some open affine in $(\Pj^n)^\ast$ containing $A'$ and $J$.  The hyperplane $\frak{B} \cap \mathbb{A}$ is the zero locus of some degree one regular function $f$ on $\mathbb{A}$.  On $\ell \cap \mathbb{A}$, this function will be nonzero at $J$ (since $J(\beta) \neq 0$), so $f$ is linear and not identically zero.  It follows that $f(A) = 0$ is the unique zero of $f$ on $\ell$, occurring with multiplicity one.  Thus, the restriction of $f$ to the local ring of $\ell$ at $A'$ is in the maximal ideal but not its square:
\be 
f \neq 0 \in \frak{m}_{\ell}/\frak{m}_{\ell}^2 = T^*_{A'}(\ell) \mbox{\hspace{15pt} where $\frak{m}_{\ell}$ denotes the maximal ideal in $ \mathcal{O}_{\ell,A'} $}.
\ee
On the other hand, Proposition \ref{tangency} shows that $(\frak{B}, A') \in W_{X^\vee}$. Consequently, $\frak{B}$ must be tangent to $S$, that is, $f$ restricted to $S$ \emph{is} in the square of the maximal ideal of the local ring of $S$ at $A'$:
\be 
f = 0 \in \frak{m}_{S}/\frak{m}_{S}^2 = T^*_{A'}(S) \mbox{\hspace{15pt} where $\frak{m}_{S}$ denotes the maximal ideal in $ \mathcal{O}_{S,A'} $}.
\ee
The function $f$ must be zero in the cotangent space of the intersection $S \cap \ell$ since the inclusion $S \cap \ell \inj S$ induces a surjection
\be T^*_{A'}(S) \surj T^*_{A'}(S \cap \ell), \ee
so the corresponding surjection
\be T^*_{A'}(\ell) \surj T^*_{A'}(S \cap \ell) \ee
must kill $f$.  This first space is the cotangent space of a line, hence one dimensional.  But $f$ is nonzero in the first space, so the second space must be zero.  It follows that $S \cap \ell$ is a zero dimensional variety.

Of course, $[\alpha : 1]$ lies in $S \cap \ell$, which is defined over $\mathbb{Q}$!  The number $\alpha$ must be algebraic.
\end{proof}

\noindent Therefore, the set of $p$-cooperative numbers is contained in $\ol\Q \cap [0,1]$, and we have established the following proposition:
\begin{prop}
If several people with finite random sources simulate a private random source for someone else, that source must take probabilities in $\ol\Q$.
\end{prop}

\subsection{Three players: what can be done}

We prove that three players with private full-strength finite random sources are enough to simulate any private finite $\ol\Q$-random source.  First, we give a construction for a hypermatrix with stochastic mystery-vectors for a given algebraic number $\alpha$, but whose entries may be negative.  Next, we use it to find a nonnegative hypermatrix with mystery-value $(\alpha+r)/s$ for some suitable natural numbers $r$ and $s$.  Then, after a bit of convex geometry to ``even out'' this hypermatrix, we scale and shift it back, completing the construction.

\begin{rmk}
Our construction may easily be made algorithmic, but in practice it gives hypermatrices that are far larger than optimal.  An optimal algorithm would need to be radically different to take full advantage of the third person.  The heart of our construction (see Proposition \ref{construction}) utilizes $2 \times (n+1) \times (n+1)$ hypermatrices, but the degree of the hyperdeterminant polynomial grows much more quickly for (near-)diagonal formats \cite[Ch. 14]{GKZ}.  We would be excited to see a method of producing (say) small cubic hypermatrices with particular mystery-values.
\end{rmk}

\subsubsection{Hypermatrices with cooperative entries}

Recall that a $\{$heads, tails$\}$-function of several finite probability spaces may be represented by a $\{1,0\}$-hypermatrix.  The condition that the entries of the matrix are either $1$ or $0$ is inconvenient when we want to build simulations for a given algebraic bias.  Fortunately, constructing a matrix with cooperative entries will suffice.

\begin{lem}\label{rationalmat}
Suppose that $A$ is a $p$-dimensional hypermatrix with p-cooperative entries in $[0,1]$ and stochastic mystery-vectors $\beeta{1}, \ldots, \beeta{p}$ for the mystery-value $\alpha$.  Then, $\alpha$ is $p$-cooperative.
\end{lem}
\begin{proof}
Let the hypermatrix $A$ have entries $w_1, w_2, \ldots , w_n $.  Each entry $w_k$ is $p$-cooperative, so it is the mystery-value of some $p$-dimensional $\{0,1\}$-hypermatrix $A_k$ with associated stochastic mystery-vectors $\beeeta{k}{1}, \beeeta{k}{2}, \ldots, \beeeta{k}{p}$.  We now build a $\{0,1\}$-hypermatrix $A'$ with $\alpha$ as a mystery-value.  The hypermatrix $A'$ has blocks corresponding to the entries of $A$.  We replace each entry $w_i$ of $A$ with a Kronecker product:
\be
w_i \mbox{ becomes } J_1 \otimes J_2 \otimes \cdots \otimes J_{i-1} \otimes A_i \otimes J_{i+1} \otimes \cdots \otimes J_n.
\ee
It is easy to check that the resulting tensor $A'$ has $\alpha$ as a mystery-value with corresponding mystery-vectors $\beeta{i} \otimes \beeeta{1}{i} \otimes \beeeta{2}{i} \otimes \cdots \otimes \beeeta{n}{i}$.
\end{proof}

\noindent Because rational numbers are $2$-cooperative, this lemma applies in particular to rational $p$-dimensional hypermatrices, for $p \geq 2$.  In this case and in others, the construction can be modified to give an $A'$ of smaller format.

%Consider a rational hypermatrix $A$ of format $k_1 \times \cdots \times k_p$, having entries between $0$ and $1$ and $\alpha$ as a mystery-value with stochastic mystery-vectors.  The construction in the above proof will produce a truly immense $\{0,1\}$-hypermatrix: The thickness of the first two dimensions will be increased by a factor of the product of the denominators of the entries.  However, there is an easy construction which increases the thickness of the first two dimensions by only a factor of the least common denominator of the entries.  If we let $d$ denote this least common denominator, then we may find (by the methods of proposition \ref{2rat}) $d \times d$ $\{0,1\}$-matrices for each entry, \emph{having identical mystery-vectors}.  There is now no need to keep these matrices away from one another's mystery-vectors by putting them in different coordinates of Kronecker products.  Replacing the entries of $A$ with these matrices already gives a hypermatrix with the desired properties.

%In the general case, one way to construct hypermatrices for every entry that have the same mystery-vectors is to express each entry in terms of a primitive element of the composite field, and then use the approximation lemma \ref{main} to smooth things out enough to satisfy the nonnegativity conditions.  In special cases, this will be more efficient than the above, but not generally.

Readers who have been following the analogy between mystery-values and eigenvalues will see that Lemma \ref{rationalmat} corresponds to an analogous result for eigenvalues of matrices.  Nonetheless, there are striking differences between the theories of mystery-values and eigenvalues.  For instance, we are in the midst of showing that it is always possible to construct a nonnegative rational hypermatrix with a given nonnegative algebraic mystery-value and stochastic mystery-vectors.  The analogous statement for matrix eigenvalues is false, by the Perron-Frobenius theorem: any such algebraic number must be greater than or equal to all of its Galois conjugates (which will also occur as eigenvalues).  Encouragingly, the inverse problem for eigenvalues has been solved: Every ``Perron number'' may be realized as a ``Perron eigenvalue'' \cite{lind}.  Our solution to the corresponding inverse problem for mystery-values uses different techniques.  It would be nice to see if either proof sheds light on the other.

%  On the other hand, the mystery-value inverse problem has no \emph{a priori} obstruction analogous to the Perron constraint.

\subsubsection{Constructing hypermatrices from matrices}

%We use $J$ to denote the all-ones tensor of any format.  When we apply $J$ to a vector or a hypermatrix, we always mean ``add up the entries.''

%\begin{defn}
%A \emph{positive (nonnegative) hypermatrix} is one with all its entries positive (nonnegative).
%\end{defn}

\begin{prop}\label{companion}
If $\lambda$ is a real algebraic number of degree $n$, then there is some $M \in \M_n(\Q)$ having $\lambda$ as an eigenvalue with non-perpendicular positive left and right eigenvectors.
\end{prop}
\begin{proof}
Let $f \in \Q[x]$ be the minimal polynomial for $\lambda$ over $\Q$, and let $L$ be the companion matrix for $f$.    That is, if
\be 
f(x) = x^n + \sum_{k=0}^{n-1} a_k x^k \mbox{ for } a_k \in \Q,
\ee
then
\be 
L=\left(
\begin{array}{ccccc}
0 & 0 & \cdots & 0 & -a_0 \\
1 & 0 & \cdots & 0 & -a_1 \\
0 & 1 & \cdots & 0 & -a_2 \\
\vdots & \vdots & \ddots & \vdots & \vdots \\
0 & 0 & \cdots & 1 & -a_{n-1}
\end{array}
\right).
\ee
The polynomial $f$ is irreducible over $\Q$, so it has no repeated roots in $\C$.  The matrix $L$ is therefore diagonalizable, with diagonal entries the roots of $f$.  Fix a basis for which $L$ is diagonal, with $\lambda$ in the upper-left entry.  In this basis, the right and left eigenvectors, $v_0$ and $w_0$, corresponding to $\lambda$ are zero except in the first coordinate.  It follows that $v_0(w_0) \neq 0$.

The right and left eigenvectors may now be visualized as two geometric objects: a real hyperplane and a real vector not contained in it.  It's clear that $\GL_n(\R)$ acts transitively on the space $\mathcal{S} := \{(v,w) \in (\R^n)^* \times \R^n : v(w) = v_0(w_0) \}$.  Moreover, $\GL_n(\Q)$ is dense in $\GL_n(\R)$, so the orbit of $(v_0,w_0)$ under the action of $\GL_n(\Q)$ is dense in $\mathcal{S}$.  The set of positive pairs in $\mathcal{S}$ is non-empty and open, so we may rationally conjugate $L$ to a basis which makes $v_0$ and $w_0$ positive.
\end{proof}

\begin{prop}\label{construction}
If $\lambda$ is real algebraic, then there exist integers $r \geq 0$, $s > 0$ such that $(\lambda+r)/s \in \coop{3}$.
\end{prop}
\begin{proof}
By Proposition \ref{companion}, there is a rational $n \times n$ matrix $M$ with non-perpendicular positive right and left eigenvectors $v,w$ for the eigenvalue $\lambda$.  Rescale $w$ so that $v(w) = 1$, and choose an integer $q \geq \max{\{J(v), J(w)\}}$.  Define the block $2 \times (n+1) \times (n+1)$ hypermatrix
\be 
A := \left(
\ba{ccc}
\ba{c|cccr}
0 \ & \ 0 & \hspace{6pt} \cdots & & 0 \\ \hline
0 \ &   & &   & \\
\vdots \ &   & \hspace{8pt}  q^2 M & & \\
0 \ &   & & &
\ea
& \rule[-34pt]{1.1pt}{74pt} &
\ba{l|cccr}
1 \ & \ 1 & & \ \hspace{8pt} \cdots & 1 \\ \hline
1 &   &   & & \\
\vdots & \multicolumn{4}{c}{q^2 (M - I) + J} \\
1 &   &   & &
\ea 
\ea \right),
\ee
where $I$ and $J$ are the $n \times n$ identity and all-ones matrices, respectively.  Consider $A$ as a trilinear form, where the metacolumns correspond to the coordinates of the first vector, the rows the second, and the columns the third.  Define the block vectors
\be
\ba{rclccccccrl}
\beeta{1} &=& \left( \hspace*{2.5pt} 1-\lambda \right. & \lambda \left. \right), & & & & & & &\\
\beeta{2} &=&
	\multicolumn{2}{l}{\left( \right. 1 - J(v)/q } & |
	& v_1 / q & v_2 / q & & \cdots & v_n / q \left. \right)& \hspace{-0pt} \mbox{, and} \\
\beeta{3} &=& \multicolumn{2}{l}{\left( \right. 1 - J(w)/q } & | & w_1 / q & w_2 / q & & \cdots & w_n / q \left. \right)& \hspace{-0pt}.
\ea
\ee
Clearly, these are all probability vectors.  It's easy to verify that
\bea
A\left(\ex{1}, \beeta{2}, \beeta{3} \right) &=& \lambda J\left(\ex{1}\right),\nn\\
A\left(\beeta{1}, \ex{2}, \beeta{3} \right) &=& \lambda J\left(\ex{2}\right), \mbox{ and}\nn\\
A\left(\beeta{1}, \beeta{2}, \ex{3} \right) &=& \lambda J\left(\ex{3}\right).
\eea
Choose a nonnegative integer $r$ large enough so that all the entries of $A + r J$ are positive, and then a positive integer $s$ so that all the entries of $A' := (A + r J)/s$ are between $0$ and $1$.
\bea
A'\left(x^{(1)}, \beeta{2}, \beeta{3}\right) &=& \frac{\lambda + r}{s} J\left(\ex{1}\right),\nn\\
A'\left(\beeta{1}, x^{(2)}, \beeta{3}\right) &=& \frac{\lambda + r}{s} J\left(\ex{2}\right), \mbox{ and}\nn\\
A'\left(\beeta{1}, \beeta{2}, x^{(3)}\right) &=& \frac{\lambda + r}{s} J\left(\ex{3}\right).
\eea
By Lemma \ref{rationalmat}, it follows that $(\lambda + r)/s$ is $3$-cooperative.
\end{proof}

\subsubsection{Finishing the Proof}

The following lemma, which we we prove later, enables us to complete the goal of this section: to classify which private random sources three or more people can simulate.

\begin{lem}[Approximation lemma]\label{approximation}
Let $\alpha$ be a $p$-cooperative number.  Now for any $\varepsilon > 0$ there exists a $p$-dimensional rational hypermatrix whose entries are all within $\varepsilon$ of $\alpha$, having $\alpha$ as a mystery-value with stochastic mystery-vectors.
\end{lem}

\begin{thm}\label{constructive}
$\coop{p} = \ol{\Q} \cap [0,1]$ for each $p \geq 3$.
\end{thm}
\begin{proof}
Certainly $0$ and $1$ are $3$-cooperative.  Let $\alpha$ be an algebraic number in $(0,1)$.  By Proposition \ref{construction}, there are integers $r \geq 0$, $s>0$ so that $(\alpha + r)/s$ is $3$-cooperative.  Let $\varepsilon := \left( \min \{\alpha, 1-\alpha\} \right) / s$.

By Proposition \ref{approximation}, there is some three-dimensional rational hypermatrix $A$ whose entries are all within $\varepsilon$ of $(\alpha + r)/s$, having $(\alpha + r)/s$ as a mystery-value with stochastic mystery-vectors.  Then, $s A - r J$ is a three-dimensional rational hypermatrix with entries between $0$ and $1$, having $\alpha$ as a mystery-value with stochastic mystery-vectors.  By Lemma \ref{rationalmat}, $\alpha$ is $3$-cooperative.

We already showed that all cooperative numbers are algebraic.  Thus, for $p \geq 3$,
\be \ol{\Q} \cap [0,1] \subseteq \coop{3} \subseteq \coop{p} \subseteq \ol{\Q} \cap [0,1], \ee
so $\coop{p} = \ol{\Q} \cap [0,1]$.
\end{proof}

%|a-(\alpha + r)/s|<e
%|(sa-r)-\alpha|<se

In conclusion, we have the following theorem.

\begin{thm}\label{main}
Three or more people with finite random sources can robustly simulate only $\ol\Q$-random sources.  Indeed, if they have full-strength finite $\ol\Q$-random sources, they can already robustly simulate a full-strength finite $\ol\Q$-random source.
\end{thm}

\subsubsection{Proof of the approximation lemma}\label{aapp}
The proof that follows is a somewhat lengthy ``delta-epsilon'' argument broken down into several smaller steps.  As we believe our construction of a hypermatrix with mystery-value $\alpha$ to be far from optimal, we strive for ease of exposition rather than focusing on achieving tight bounds at each step along the way.

Recall that a finite probability space may be usefully modeled by a positive\footnote{We may leave out points of mass zero.} vector of mass one.  Let $\beta$ be such a vector.  We denote by $\#\beta$ the number of coordinates of $\beta$ .  We say $\beta^\p$ is a \emph{refinement} of $\beta$ when $\beta$ is the image of a measure-preserving map from $\beta^\p$; that is, when the coordinates of $\beta^\p$ may be obtained by splitting up the coordinates of $\beta$.

The following easy lemma states that any positive vector of unit mass can be refined in such a way that all the coordinates are about the same size.

\begin{lem}[Refinement lemma]
Let $\beta$ be a positive vector of total mass $1$.  For any $\delta > 0$ there exists a refinement $\beta^\p$ of $\beta$ with the property that
\be \min_j \beta_j^\p \geq \frac{1-\delta}{\#\beta^\p}. \ee
\end{lem}
\begin{proof}
Without loss of generality, assume that $\beta_1$ is the smallest coordinate of $\beta$.  Let $\gamma = \beta_1 \delta$, and let $k = \#\beta$.  The vector $\beta$ is in the standard open $k$-simplex
\be 
\Delta^k = \{\mbox{positive vectors of mass 1 and dimension $k$}\}.
\ee
The rational points in $\Delta^k$ are dense (as in any rational polytope), and
\be 
U := \{x \in \Delta^k : (\forall i) \left|\beta_i -x_i\right| < \gamma \mbox{ and } \beta_1 < x_1\}
\ee
is an open subset of the simplex.  So $U$ contain a rational point $\left(\frac{n_1}{n}, \ldots, \frac{n_k}{n}\right)$, with $n = \sum n_i$.  Thus, $\left|\beta_i - \frac{n_i}{n}\right| < \gamma$ and $\beta_1 < \frac{n_1}{n}$, so
\be 
\left|\frac{\beta_i}{n_i} - \frac{1}{n}\right| < \frac{\gamma}{n_i} \leq \frac{\gamma}{n_1} < \frac{\gamma}{\beta_j n} = \frac{\delta}{n}.
\ee
Let $\beta^\p$ be the refinement of $\beta$ obtained by splitting up $\beta_i$ into $n_i$ equal-sized pieces.  We have $\#\beta^\p = n$, and the claim follows from this last inequality.
\end{proof}

\begin{rmk}
The best general bounds on the smallest possible $\#\beta^\p$ given $\beta$ and $\delta$ are not generally known, but fairly good bounds may be obtained from the multidimensional version of Dirichlet's theorem on rational approximation, which is classical and elementary \cite{davenport}.  Actually calculating good simultaneous rational approximations is a difficult problem, and one wishing to make an algorithmic version of our construction should consult the literature on multidimensional continued fractions and Farey partitions, for example, \cite{lagarias, farey}.
\end{rmk}

\noindent The next proposition is rather geometrical.  It concerns the $n \times n$ matrix $S_{\delta} := (1-\delta)(\J/n ) + \delta \I$, which is a convex combination of two maps on the standard simplex: the averaging map and the identity map.  Each vertex gets mapped almost to the center, so the action of $S_{\delta}$ can be visualized as shrinking the standard simplex around its center point.  The proposition picks up where the refinement lemma left off:

\begin{prop}\label{second}
If a stochastic vector $\beta$ satisfies
\be \min_i \beta_i \geq \frac{1-\delta}{\#\beta} \ee then its image under the map $S^{-1}_{\delta} $ is still stochastic.
\end{prop}
\begin{proof}
First note that
$\left[ (1-\delta) \left(J / \# \beta \right) + \delta \I \right] \left[ (1-1/\delta)\left(J / \# \beta \right) + (1/\delta) \I \right] = \I$,
so we have an explicit form for $S^{-1}_{\delta}$.
We know that $ \min_i \beta_i \geq (1-\delta)/\#\beta$, so the vector \be E=\frac{1}{\delta} \left[ \beta - \left( \frac{1-\delta}{\#\beta} \right) \J \right] \ee
is still positive.  Now
$\beta = (1-\delta) \left(\J / \#\beta \right) + \delta E$, 
a convex combination of two positive vectors.  The vector $\beta$ has mass $1$, and $ \left(\J / \#\beta \right)$ as well, so $E$ also has mass $1$.

Now compute:
\begin{eqnarray}
S^{-1}_{\delta} \beta & = & \biggl[ (1-1/\delta) \left(\J / \#\beta \right) + (1/\delta) \I \biggr] \biggl[ (1-\delta) \left(\J / \#\beta \right) + \delta E \biggr] \nn \\
 & = & \biggl[ (1 - 1/\delta)(1 - \delta) + (1/\delta)(1 - \delta) + (1 - 1/\delta)\delta \biggr] \left(\J / \#\beta \right) + E \nn \\
 & = & E.
\end{eqnarray}
This completes the proof.
\end{proof}

\noindent The following proposition shows that applying the matrix $S_{\delta}$ in all arguments of some multilinear functional forces the outputs to be close to each other.
\begin{prop}\label{third}
Let $A$ be a hypermatrix of format $n_1 \times n_2 \times \cdots \times n_p$ with entries in $[0,1]$, and take $\delta := \varepsilon / (2p)$.  Now the matrix $A'$ defined by
\be A^\p\left(\otimes x^{(i)}\right) := A\left(\otimes S_{\delta} x^{(i)}\right) \ee
satisfies $| A'(x) - A'(x') | \leq \varepsilon$
for any two stochastic tensors $x$ and $x'$.
\end{prop}
\begin{proof}
Let $m := A\left(\otimes (\J / n_i)\right)$, the mean of the entries of $A$.  We show that for any stochastic vectors $x^{(i)}$,
\be \left|A'\left(\otimes x^{(i)}\right) - m\right| \leq \varepsilon/2. \ee
Since any other stochastic tensor is a convex combination of stochastic pure tensors, it will follow that
$|A'(x) - m| \leq \varepsilon/2$.
Then the triangle inequality will yield the result.

It remains to show that $A'$ applied to a stochastic pure tensor gives a value within $\varepsilon /2 $ of $m$.
\begin{eqnarray}
A'\left(\otimes x^{(i)}\right) & = & A\left(\otimes S_{\delta} x^{(i)}\right) \nn \\
 & = & A\left( \otimes \left[ (1-\delta)(\J /n_i) + \delta \I \right] x^{(i)} \right) \nn \\
 & = & A\left( \otimes \left[(1-\delta)(\J /n_i) + \delta x^{(i)} \right] \right). \label{bad} %\nn \\
% & = & (1-\delta)^p A\left(\otimes (\J / n_i)\right) + \mbox{the rest of the convex combination} \nn \\
% & = & (1-\delta)^p m + \mbox{the rest of the convex combination} \nn \\
% & = & (1-\varepsilon/2) m + (\varepsilon/2) (\mbox{some element of } [0,1])
\end{eqnarray}
Each argument of $A$---that is, factor in the tensor product---is a convex combination of two stochastic vectors.  Expanding out by multilinearity, we get convex combination with $2^p$ points.  Each point---let's call the $k^{\mbox{\tiny{th}}}$ one $y_k$---is an element of $[0,1]$ since it is some weighted average of the entries of $A$.  This convex combination has positive $\mu_k$ such that $\sum \mu_k = 1$ and
\be A'\left(\otimes x^{(i)}\right) = \sum_{k=1}^{2^p} \mu_k y_k. \ee
Taking the first vector in each argument of $A$ in \eqref{bad}, we see that $y_1 = A\left(\otimes (\J / n_i)\right)=m$, the average entry of $A$.  Thus, the first term in the convex combination is
$\mu_1 y_1 = (1-\delta)^p m$.

The inequality $(1-\varepsilon/2) \leq (1-\delta)^p$ allows us to split up the first term.  Let $\mu_0 := 1-\varepsilon/2$ and $\mu_1' := \mu_1 - \mu_0 \geq 0$.  We have
$\mu_1 y_1 = (\mu_0 + \mu_1') y_1 = (1-\varepsilon/2) m + \mu_1' m$.
After splitting this term, the original convex combination becomes
\be 
A'\left(\otimes x^{(i)}\right) = (1-\varepsilon/2) m + \mu_1' m + \sum_{k=2}^{2^p} \mu_k y_k.
\ee
Let $e$ denote the weighted average of the terms after the first.  We may rewrite the convex combination
\be A'\left(\otimes x^{(i)}\right) = (1-\varepsilon/2) m + (\varepsilon/2) e. \ee
Since $m, e \in [0,1]$,
\be m-\varepsilon / 2 \leq (1-\varepsilon /2) m \leq A'\left(\otimes x^{(i)}\right) \leq (1-\varepsilon /2) m + \varepsilon / 2 \leq m + \varepsilon /2, \ee
and
\be \left| A'\left(\otimes x^{(i)}\right) - m \right| \leq \varepsilon / 2,\ee
so we are done.
\end{proof}

These results are now strong enough to prove the approximation lemma \ref{approximation}.
%\begin{lem}[Approximation lemma]
%Let $\alpha$ be a $p$-cooperative number.  Now for any $\varepsilon > 0$ there exists a $p$-format rational hypermatrix whose entries are all within $\varepsilon$ of $\alpha$ having $\alpha$ as a mystery-value with stochastic mystery-vectors.
%\end{lem}
\begin{proof}
The number $\alpha$ is $p$-cooperative, so it comes with some $p$-dimensional nonnegative rational hypermatrix $A$ and positive vectors $\beeta{1}, \beeta{2}, \ldots, \beeta{p}$ of mass one, satisfying (in particular) $A\left( \otimes \beeta{i} \right) = \alpha$.  The refinement lemma allows us to assume that each $\beeta{i}$ satisfies 
\be \min_j \beta^{(i)}_j \geq \frac{1-\delta}{\#\beta^{(i)}}. \ee
If one of the $\beeta{i}$ fails to satisfy this hypothesis, we may replace it with the refinement given by the lemma, and duplicate the corresponding slices in $A$ to match.

Now, by Proposition \ref{second}, each $S^{-1}_{\delta} \beeta{i}$ is a stochastic vector.

Let $A'$ be as in Proposition \ref{third}.  It will still be a rational hypermatrix if we pick $\varepsilon$ to be rational.  We know
\be A'\left( \otimes S^{-1}_{\delta} \beeta{i} \right) = \alpha. \ee
On the other hand, any entry of the matrix $A'$ is given by evaluation at a tensor product of basis vectors.  Both $\alpha$ and any entry of $A'$ can be found by evaluating $A'$ at a stochastic tensor.  Thus, by Proposition \ref{third}, each entry of $A'$ is within $\varepsilon$ of $\alpha$.
\end{proof}

\subsection{Higher-order robustness}

We complete the proof of our main theorem.

\begin{prop}\label{onlyrat}
If $r \geq p/2$, then $p$ people with finite random sources may $r$-robustly simulate only finite $\Q$-random sources.
\end{prop}
\begin{proof}
Consider an $r$-robust simulation.  Imagine that Alice has access to half of the random sources (say, rounded up), and Bob has access to the remaining sources.  Because Alice and Bob have access to no more than $r$ random sources, neither knows anything about the source being simulated.  But this is precisely the two-player case of ordinary $1$-robustness, so the source being simulated is restricted to rational probabilities.
\end{proof}

%\begin{prop}\label{crappy}
%Let $n \in \Z^+$.  If $3^n$ people have full-strength finite $\ol\Q$-random sources, they may $(2^n-1)$-robustly simulate any finite $\ol\Q$-random source.
%\end{prop}
%\begin{proof}
%For nine people to $3$-robustly simulate a finite $\ol\Q$-random source, they split into threes, and they $1$-robustly simulate three independent random sources $1$-robustly simulating one random source, which is possible by theorem \ref{main}.  In general, we use induction.  The $3^n$ people split into three groups of $3^{n-1}$, and they $(2^{n-1}-1)$-robustly simulate three independent random sources $1$-robustly simulating one random source.
%\end{proof}

%The existence of $q$-robust simulation of $\ol \Q$-random sources is a trickier question.  It can be solved by simulating simulations (and simulating simulations of simulations, etc.).  We address this in detail in a companion paper, proving the following theorem.

\noindent In the constructive direction, we show the following:

\begin{prop}\label{bigcon}
If $r < p/2$, then $p$ people with full-strength finite $\ol \Q$-random sources may $r$-robustly simulate a full-strength finite $\ol \Q$-random source.
\end{prop}

\noindent The proof is to simulate simulations (and simulate simulations of simulations, etc.).  We treat the $p=3$ case of our $1$-robust simulation protocol as a black box.  If a majority of the random sources put into it are reliable, the one that comes out (the simulated random source) will also be reliable.  This viewpoint leads us into a discussion of majority gates.

\begin{defn}
A \textbf{p-ary majority gate} is a logic gate that computes a boolean function returning $1$ if a majority of its inputs are $1$ and $0$ if a majority of its inputs are $0$.  (The output doesn't matter when there are ties.)
\end{defn}

\begin{lem}[Bureaucracy]
A $p$-ary majority gate may be built by wiring together ternary majority gates.
\end{lem}

\noindent The proof of the bureaucracy lemma is a straightforward application of the probabilistic method, and is covered in detail in Appendix \ref{bapp}.  Now, by iterating simulations of simulations according to the wiring provided by the bureaucracy lemma, we can overcome any minority of malfunctioning sources.  So the bureaucracy lemma, together with the ``black box'' of our three-player construction, implies Proposition \ref{bigcon}.

Now we're finally ready to prove our main result.  The statement here is equivalent to the ones in the abstract and in Section \ref{results} but uses the language of robustness.

\begin{thm}\label{main}
Say $p$ people have full-strength finite random sources.  If $p/2 \leq r < p$, the people may $r$-robustly simulate any finite $\mathbb{Q}$-random source and nothing better; if $1 \leq r < p/2$, they may $r$-robustly simulate any finite $\overline{\mathbb{Q}}$-random source and nothing better.
\end{thm}
\begin{proof}
The claim simply combines Proposition \ref{ratcon}, Proposition \ref{onlyrat}, Theorem \ref{main}, and Proposition \ref{bigcon}.
\end{proof}

\section{Application to Secure Multiparty Computation and Mental Poker}\label{application}

We begin with the classical case: Three gentlemen wish to play poker, but they live far away from each other, so playing with actual cards is out of the question.  They could play online poker, in which another party (the remotely hosted poker program) acts as a dealer and moderator, keeping track of the cards in each player's hand, in the deck, etc., and giving each player exactly the information he would receive in a physical game.  But this solution require our gentlemen to trust the moderator!  If they fear the moderator may favor one of them, or if they wish to keep their game and its outcome private, they need another system.

A better solution is to use secure multiparty computation.  Our gentlemen work to \emph{simulate} a moderator in a way that keeps the outcomes of the moderator's computations completely hidden from each of them.  An unconditionally-secure method of playing poker (and running other games/computations) ``over the phone'' has been described in \cite{poker}.

In the classical case, the players may perform finite computations, communicate along private channels, and query full-strength finitary private random sources.  The simulated moderator has the almost same abilities as the players, except that its private random source is limited to rational probabilities.  The work of this paper expands this to all algebraic probabilities, and shows that one can do no better.

To see how this may be useful, think back to our poker players.  They may be preparing for a poker tournament, and they may want to simulate opponents who employ certain betting strategies.  But poker is a complicated multiplayer game (in the sense of economic game theory), and Nash equilibria will occur at mixed strategies with algebraic coefficients.\footnote{The appearance of algebraic (but not transcendental) coefficients in mixed strategies is explained by R. J. Lipton and E. Markakis here \cite{nash}.}

\appendix

\section{Relevant Constructions in Algebraic Geometry}\label{geoapp}

Comprehensive introductions to these constructions may be found in \cite[Lec. 14, 15, 16]{harris} and \cite[Ch. 1]{GKZ}.

\subsection{Tangency and projective duality}\label{geoapp1}

Let $k$ be an algebraically closed field of characteristic zero.  (For our purposes, it would suffice to take $k=\C$, but the methods are completely general.)  Let $X \subseteq \Pj^n$ be a projective variety over $k$.  A hyperplane $H \in (\Pj^n)^\ast$ is \emph{(algebraically) tangent} to $X$ at a point $z$ if every regular function on an affine neighborhood of $z$ vanishing on $H$ lies in the square of the maximal ideal of the local ring $\mathcal{O}_{X,z}$.

This notion of tangency agrees with geometric intuition on the set of smooth points $X^\sm$ of $X$.  To get a more complete geometric picture, we define an incidence variety:
\be 
W_X := \ol{\{ (z,H) : z \in X^\sm, H \mbox{ is tangent to } X \mbox{ at } z \}} \subseteq \Pj^n \times (\Pj^n)^\ast.
\ee
The bar denotes Zariski closure.  Membership in $W_X$ may be thought of as extending the notion of tangency at a smooth point to include singular points ``by continuity.''

The image of a projective variety under a regular map is Zariski closed, so the projection of $W_X$ onto the second coordinate is a variety, called the dual variety and denoted $X^{\vee}$.

The following theorem explains why projective duality is called ``duality.''  We omit the proof; see \cite[p. 208--209]{harris} or \cite[p. 27--30]{GKZ}.

\begin{thm}[Biduality theorem]\label{biduality}
Let $X$ be a variety in $\Pj^n$.  For $z \in \Pj^n$, let $z^{\ast\ast}$ be the image under the natural isomorphism to $(\Pj^n)^{\ast\ast}$.  Then, $(z,H) \mapsto (H,z^{\ast\ast})$ defines an isomorphism $W_X \isom W_{X^\vee}$.  (Specializing to the case when $(z,H)$ and $(H,z^{\ast\ast}$ are smooth points $X$ and $X^\vee$, respectively, this says that $H$ is tangent to $X$ at $z$ if and only if $z$ is tangent to $X^{\vee}$ at $H$.)  Moreover, $z \mapsto z^{\ast\ast}$ defines an isomorphism $X \isom (X^{\vee})^{\vee}$.
\end{thm}

\subsection{Segre embeddings and their duals}\label{format}

Consider the natural map $k^{n_1} \times \cdots \times k^{n_p} \to k^{n_1} \tensor \cdots \tensor k^{n_p} = k^{n_1 \cdots n_p}$ given by the tensor product.  Under this map, the fiber of a line through the origin is a tuple of lines through the origin.  Thus, this map induces an embedding $\Pj^{n_1-1} \times \cdots \times \Pj^{n_p-1} \inj \Pj^{n_1 \cdots n_p - 1}$.  The map is known as the Segre embedding, and the image is known as the Segre variety $X$ of \emph{format} $n_1 \times \cdots \times n_p$.  It is, in other words, the pure tensors considered as a subvariety of all tensors, up to constant multiples.  This variety is cut out by the determinants of the $2 \times 2$ subblocks.  Also, it is smooth because it is isomorphic as a variety to $\Pj^{n_1-1} \times \cdots \times \Pj^{n_p-1}$.

When a projective variety is defined over the rational numbers,\footnote{That is, it is the zero set of a system of homogeneous rational polynomials.} its dual is also defined over the rationals, by construction \cite[p. 14]{GKZ}.  In particular, the dual $X^\vee$ of the Segre embedding is defined over $\Q$.

When the dimensions $n_i$ satisfy the ``$p$-gon inequality''
\be 
(n_j-1) \leq \sum_{i \neq j} (n_i-1),
\ee
\noindent Gelfand, Kapranov, and Zelevinsky \cite[p. 446]{GKZ} show that the dual of the Segre variety is a hypersurface.  The polynomial for this hypersurface is irreducible, has integer coefficients, and is known as the \emph{hyperdeterminant} of format $n_1 \times \cdots \times n_p$.  It is denoted by $\Det$.  When $p=2$ and $n_1=n_2$, the hyperdeterminant is the same as the determinant of a square matrix \cite[p. 36]{GKZ}.

Gelfand, Kapranov, and Zelevinsky provide us with two equivalent definitions of degeneracy.

\begin{defn}\label{pderiv}
A $p$-linear form $T$ is said to be \textbf{degenerate} if either of the following equivalent conditions holds:
\begin{itemize}
\item there exist nonzero vectors $\beeta{i}$ so that, for any $0 \leq j \leq p$,
\be T \left( \beeta{1}, \ldots , \beeta{j-1} , \ex{j} ,\beeta{j+1}, \ldots , \beeta{p} \right) = 0 \mbox{ for all $\ex{j}$;} \ee

\item there exist nonzero vectors $\beeta{i}$ so that $T$ vanishes at $\otimes \beeta{i}$ along with every partial derivative with respect to an entry of some $\beeta{i}$:
\be T \mbox{ and } \frac{\partial T}{\partial \beeeta{j}{i}} \mbox{ vanish at $\otimes \beeta{i}$.}\ee
\end{itemize}
\end{defn}

\noindent The dual of the Segre variety is useful to us because it can tell whether a multilinear form is degenerate.

\begin{thm}[Gelfand, Kapranov, and Zelevinsky]
For any format, the dual $X^\vee$ of the Segre embedding is defined over $\Q$ and satisfies, for every multilinear form $T$ of that format,
\be T \in X^\vee \hspace{.15in} \Longleftrightarrow \hspace{.15in} T \mbox{ is degenerate.}\ee
When the format satisfies the ``p-gon inequality,'' $X^\vee$ is defined by a polynomial in the entries of $T$ with coefficients in $\Z$, called the hyperdeterminant:
\be \Det(T) = 0 \hspace{.15in} \Longleftrightarrow \hspace{.15in} T \mbox{ is degenerate.}\ee
\end{thm}

\section{Proof of the bureaucracy lemma}\label{bapp}

Here, we show that a $p$-ary majority gate may be built out of ternary majority gates.

\begin{proof}
%e in fact prove something stronger: A random $p$-ary gate built out of ternary gates in a certain way will be a majority gate.  
We prove the existence of the majority gate by showing that a random gate built in a certain way has a positive probability of being a majority gate.  For simplicity, we assume $p$ is odd.  The even case follows from the odd case: A $(2k-1)$-ary majority gate functions as a $(2k)$-ary majority gate if we simply ignore one of the inputs.

Make a balanced ternary tree of depth $n$ out of $3^0 + 3^1 + \cdots + 3^{n-1}$ ternary majority gates, where $n$ is to be specified later.  Let $S$ be the set of possible assignments of $p$ colors (one for each input slot) to the $3^n$ leaves of the tree.  Each $s \in S$ defines a $p$-ary gate; we prove that, for $n$ large enough, a positive fraction of these are majority gates.  Let $T$ be the set of $p$-tuples of input values with exactly $\frac{p+1}{2}$ coordinates equal to 1.  For $(s,t) \in S \times T$, let $\chi(s,t)$ be the bit returned by the gate defined by $s$ on input $t$.

If each input of a $3$-ary majority gate is chosen to be $1$ with probability $x$, and $0$ with probability $1-x$, we may compute the probability $f(x)$ that the resulting bit is $1$:
\be 
f(x) = {3 \choose 2}x^2(1-x) + {3 \choose 3}x^3 = x^2(3-2x).
\ee
Fixing the choice of $t \in T$ and letting $s$ vary uniformly, it's as if we're assigning $1$ or $0$ to each leaf with probabilities $\frac{p+1}{2}$ and $\frac{p-1}{2}$, respectively.  We have
\be 
\frac{1}{|S|}\sum_{s \in S} \chi(s,t) = f^n\left(\frac{p+1}{2}\right),
\ee
where $f^n$ denotes iterated composition.  Whenever $\frac{1}{2} < \xi \leq 1$, it's easy to see that $f^n(\xi)$ approaches $1$ as $n$ becomes large.\footnote{In fact, the convergence is very fast.  While we're ignoring computational complexity questions in this paper, more careful bookkeeping shows that this proof gives a polynomial bound (in $p$) on the size of the tree.}  Choose $n$ so that $f^n\left(\frac{p+1}{2}\right) > 1 - \frac{1}{|T|}$.  Now,
\bea
\frac{1}{|S|}\sum_{s \in S} \sum_{t \in T} \chi(s,t) &=& \sum_{t \in T} \frac{1}{|S|}\sum_{s \in S} \chi(s,t)\nn\\
                                                     &=& \sum_{t \in T} f^n\left(\frac{p+1}{2}\right)\nn\\
                                                     &=& |T| f^n\left(\frac{p+1}{2}\right)\nn\\
                                                     &>& |T| \left(1 - \frac{1}{|T|}\right) = |T|-1.
\eea
This is an average over $S$, and it follows that there must be some particular $s_0 \in S$ so that the inner sum $\sum_{t \in T} \chi(s_0,t)$ is greater than $|T|-1$.  But that sum clearly takes an integer value between $0$ and $|T|$, so it must take the value $|T|$, and we have $\chi(s_0,t) = 1$ for every $t \in T$.  That is, the gate specified by $s_0$ returns $1$ whenever exactly $\frac{p+1}{2}$ of the inputs are $1$.  By construction, setting more inputs to $1$ will not alter this outcome, so the gate returns $1$ whenever a majority of the inputs are $1$.  By the symmetry between $1$ and $0$ in each ternary component, the gate returns $0$ whenever a majority of the inputs are $0$.  Thus, $s_0$ defines a $p$-ary majority gate.
\end{proof}

\noindent We illustrate a $5$-ary majority gate of the type obtained in the bureaucracy lemma:
\begin{center} \includegraphics[scale=1.4]{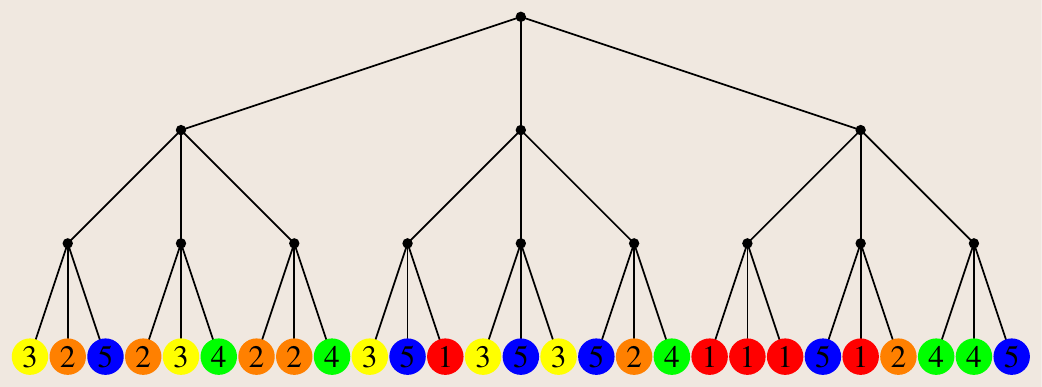} \end{center}

\section{Simulating infinite random sources}\label{infapp}
Say Alice and Bob are both equipped with private, full-strength random sources; they wish to simulate a private, full-strength random source for some other player.

For technical reasons, we will take ``full-strength random source'' to mean ``a random source capable of sampling from any Haar measure.''  This restriction is mostly to avoid venturing into the wilds of set theory.  After all, the pathologies available to probability spaces closely reflect the chosen set-theoretic axioms.  We call these restricted spaces ``Haar spaces.''

\begin{defn} A probability space $P$ is a \textbf{Haar space} if there exists some compact topological group $G$, equipped with its normalized Haar measure, admitting a measure-preserving map to $P$.
\end{defn}

\begin{rmk}
The following probability spaces are all Haar spaces:
any continuous distribution on the real line;
any standard probability space in the sense of Rokhlin \cite{rokhlin};
any Borel space or Borel measure on a Polish space;
any finite probability space;
arbitrary products of the above.
\end{rmk}

\noindent The following construction is an easy generalization of the classical construction given in Proposition \ref{diecon}.

\begin{prop}\label{robustgroups}\label{groupconstruction}
Let $G$ be a compact group with normalized Haar measure.  Now, $p$ players equipped with private sources that sample from $G$ may $(p-1)$-robustly simulate an source that samples from $G$.
\end{prop}
\begin{proof}
We provide a direct construction.  The $i^{\rm th}$ player uses the Haar measure to pick $g_i \in G$ at random.  The output of the simulated source will be the product $g_1g_2 \cdots g_p$.

It follows from the invariance of the Haar measure that any $p$-subset of \be \{g_1, g_2,...,g_p, g_1g_2 \cdots g_p\} \ee is independent!  Thus, this is a $(p-1)$-robust simulation.
\end{proof}

%\begin{prop}\label{groupconstruction}
%Let $G$ be a compact group with normalized Haar measure.  Two players equipped with private sources that sample from $G$ may simulate a private source that samples from $G$.
%\end{prop}
%\begin{proof}
%We provide a direct construction.  The Haar measure on $G$ is a quotient of the Haar measure on $G \times G$, via the multiplication map.  Alice and Bob independently pick elements $g_1, g_2 \in G$ at random.  They send their elements to Carol, who multiplies them.
%
%Information-theoretic security is guaranteed by the fact that $g_1$, $g_2$ and $g_1g_2$ are pairwise independent random variables.  This fact follows from the invariance of the Haar measure on $G$.
%\end{proof}

%For an example of this construction, consider how Alice and Bob may robustly flip a coin with bias $2/5$.  Alice picks an element $a \in \Z/5\Z$, and Bob picks an element $b \in \Z/5\Z$; both do so using the uniform distribution.  Then, $a,b,$ and $a+b$ are pairwise independent!  We say that the coin came up heads if $a+b \in \{0,1\}$ and tails if $a+b \in \{2,3,4\}$.

%This construction exploits the fact that several random variables may be pairwise (or $(p-1)$-setwise) independent but still dependent overall.  In cryptology, this approach goes back to the one-time pad.  Shamir \cite{shamir} uses it to develop secret-sharing protocols---essentially as we describe here---and these are exploited in multiparty computation to such ends as playing poker without cards \cite{game}.

\begin{cor}
If $p$ players are equipped with private, full-strength random sources, they may $(p-1)$-robustly simulate may simulate a private, full-strength random source for some other player.
\end{cor}
\begin{proof}
By Proposition \ref{groupconstruction}, they may simulate a private random source capable of sampling from any compact group with Haar measure.  But such a random source may also sample from all quotients of such spaces.
\end{proof}

\begin{cor}
If $p$ players are equipped with private random sources capable of sampling from the unit interval, they may $(p-1)$-robustly simulate a random source capable of sampling from any standard probability space---in particular, any finite probability space.
\end{cor}
\begin{proof}
Immediate from Proposition \ref{groupconstruction}.
\end{proof}

\bibliographystyle{alpha}{}
\bibliography{references}

\begin{thebibliography}{GMW87}

\bibitem[Cay45]{cayley}
A.~Cayley.
\newblock On the theory of linear transformations.
\newblock {\em Cambridge Math. J.}, 4, 1845.

\bibitem[Dav54]{davenport}
H.~Davenport.
\newblock Simultaneous diophantine approximation.
\newblock In {\em Proc. of ICM}, volume~3, pages 9--12, 1954.

\bibitem[GKZ94]{GKZ}
I.M. Gelfand, M.M. Kapranov, and A.V. Zelevinsky.
\newblock {\em Discriminants, Resultants and Multidimensional Determinants}.
\newblock Birkh$\mathrm{\ddot{a}}$user, 1994.

\bibitem[GM82]{poker}
S.~Goldwasser and S.~Micali.
\newblock Probabilistic encryption \& how to play mental poker keeping secret
  all partial information.
\newblock In {\em Proc. of ACM symp. on TOC}, 1982.

\bibitem[GMW87]{game}
O.~Goldreich, S.~Micali, and A.~Wigderson.
\newblock How to play {ANY} mental game.
\newblock In {\em Proc. of ACM symp. on TOC}, pages 218--229. ACM Press, 1987.

\bibitem[Har92]{harris}
J.~Harris.
\newblock {\em Algebraic Geomerty: A First Course}.
\newblock Springer, 1992.

\bibitem[Lag82]{lagarias}
J.~C. Lagarias.
\newblock Best simultaneously diophantine approximations. {I}. growth rates of
  best approximation denominators.
\newblock {\em Trans. of AMS}, 272(2):545--554, 1982.

\bibitem[Lin84]{lind}
D.~A. Lind.
\newblock The entropies of topological {M}arkov shifts and a related class of
  algebraic integers.
\newblock {\em Ergodic Theory Dynam. Systems}, 4(2):283--300, 1984.

\bibitem[LM04]{nash}
R.~J. Lipton and E.~Markakis.
\newblock Nash equilibria via polynomial equations.
\newblock In {\em LATIN'04}, pages 413--422, 2004.

\bibitem[Mum95]{mumford}
D.~Mumford.
\newblock {\em Algebraic Geometry I: Complex Projective Varieties}.
\newblock Springer, 1995.

\bibitem[NS06]{farey}
A.~Nogueira and B.~Sevennec.
\newblock Multidimensional farey partitions.
\newblock {\em Indag. Mathem.}, 17(3):437--456, 2006.

\bibitem[Rok49]{rokhlin}
V.~A. Rokhlin.
\newblock On the fundamental ideas of measure theory.
\newblock {\em Mat. Sbornik N.S.}, 25(67):107--150, 1949.

\bibitem[Sha79]{shamir}
A.~Shamir.
\newblock How to share a secret.
\newblock {\em CACM}, 22(11):612--613, 1979.

\bibitem[Yao82]{yao}
A.~C. Yao.
\newblock Protocols for secure computations (extended abstract).
\newblock In {\em Proc. of FOCS}, pages 160--164, November 1982.

\end{thebibliography}

\end{document}